\documentclass[11pt]{article}
\usepackage{amsmath,amssymb,amsthm,fullpage,verbatim,bm}
\allowdisplaybreaks
\DeclareFontFamily{U}{mathx}{\hyphenchar\font45}
\DeclareFontShape{U}{mathx}{m}{n}{
      <5> <6> <7> <8> <9> <10>
      <10.95> <12> <14.4> <17.28> <20.74> <24.88>
      mathx10
      }{}
\DeclareSymbolFont{mathx}{U}{mathx}{m}{n}
\DeclareMathSymbol{\bigplus}{1}{mathx}{"90}
\DeclareMathSymbol{\bigtimes}{1}{mathx}{"91}

\newtheorem{definition}{Definition}

\newtheorem{proposition}{Proposition}
\newtheorem{fact}{Fact}
\newtheorem{theorem}{Theorem}

\newcommand{\C}{\mathbb{C}}

\DeclareMathOperator*{\E}{{\rm {\bf E}}\,}
\DeclareMathOperator*{\Tr}{{\rm Tr}\;}

\newcommand{\zero}{\leavevmode\hbox{\small l\kern-3.5pt\normalsize0}}
\newcommand{\one}{\leavevmode\hbox{\small1\kern-3.8pt\normalsize1}}

\newcommand{\elltwo}[1]{\left\|{ #1 }\right\|_2}

\newcommand{\ket}[1]{| #1 \rangle}
\newcommand{\bra}[1]{\langle #1 |}
\newcommand{\ketbra}[1]{\ket{#1}\bra{#1}}
\newcommand{\braket}[2]{\langle {#1} \ket{#2}}

\newcommand{\polylog}{\mathrm{polylog}}
\newcommand{\poly}{\mathrm{poly}}

\title{{\bf A quantum Johnson-Lindenstrauss lemma via
unitary $t$-designs
}}

\author{
Pranab Sen\thanks{
School of Technology and Computer Science, Tata Institute of Fundamental
Research, Mumbai 400005, India.
Email: {\sf pgdsen@tcs.tifr.res.in}
}
}

\date{}

\begin{document}
\maketitle

\begin{abstract}
The famous Johnson-Lindenstrauss lemma~\cite{JL} states that for any 
set of $n$ vectors $\{v_i\}_{i=1}^n \in \C^{d_1}$ and 
any $\epsilon > 0$, there is a 
linear transformation $T: \C^{d_1} \rightarrow \C^{d_2}$, 
$d_2 = O(\epsilon ^{-2} \log n)$ such that 
$\elltwo{T(v_i)} \in (1 \pm \epsilon) \elltwo{v_i}$ for all $i \in [n]$. 
In fact, a Haar random $d_1 \times d_1$ unitary transformation followed
by projection onto the first $d_2$ coordinates followed by a scaling
of $\sqrt{\frac{d_1}{d_2}}$ works as a valid transformation 
$T$ with high probability. 
In this work, we show that the Haar random $d_1 \times d_1$ unitary
can be replaced by a uniformly random unitary chosen from a finite 
set called an approximate unitary $t$-design for $t = O(d_2)$. 
Choosing a unitary from such a design requires only 
$O(d_2 \log d_1)$ random bits as opposed to 
$2^{\Omega(d_1^2)}$ random bits required to choose a Haar random unitary
with reasonable precision.
Moreover, since such unitaries can be efficiently implemented in 
the superpositional
setting, our result can be viewed as an efficient quantum 
Johnson-Lindenstrauss transform akin to efficient quantum 
Fourier transforms
widely used in earlier work on quantum algorithms.

We prove our result by leveraging a method of
Low~\cite{Low} for showing concentration for approximate unitary $t$-designs. We
discuss algorithmic advantages and limitations of our result and conclude 
with a toy application to private information retrieval.
\end{abstract}

\section{Introduction}
The Johnson Lindenstrauss lemma is one of the oldest dimensionality
reduction results for the $\ell_2$-norm and has applications to many
problems in computer science, signal processing, compressed sensing etc. 
Informally speaking,
it says that any set of $n$ points in high dimensional Euclidean space
(say of dimension $d_1$)
can be embedded into $d_2 := O(\epsilon^{-2} \log n)$-dimensional Euclidean
space preserving all the ${n \choose 2}$ pairwise distances to within
a multiplicative factor of $1 \pm \epsilon$. An equivalent description
would be that the embedding approximately preserves all the pairwise
angles or inner products. Moreover, with high probablity
this embedding
can be achieved by taking a Haar random $d_1 \times d_1$ unitary $U$,
applying $U$ to all the points in the set, projecting onto the
first $d_2$ coordinates and scaling the result by $\sqrt{\frac{d_1}{d_2}}$.
The advantages of such an embedding are manifold: the embedding is
linear, oblivious to the actual set of points, with target dimension 
independent of the source dimension, and can be implemented
by a randomised algorithm in $O(d_1^2 \polylog (d_1))$ time.
Fast Johnson Lindenstrauss transforms, akin to fast Fourier transforms
arising from the discrete Fourier transform, have also been discovered
(see e.g.~\cite{AC}).
They typically run in $O(d_1 \polylog(d_1))$ time.

In this paper, we work in the quantum superpositional setting. By this
we mean that our source vectors are not provided explicitly, 
but rather are the state vectors of pure quantum states with Hilbert
space $\C^{d_1}$. Then, if we choose a Haar random $d_1 \times d_1$
unitary $U$, applying it via a quantum circuit to a pure state, measure
the name of a {\em block}, where the $d_1$ coordinates are divided into
$d_1/d_2$ blocks of $d_2$ coordinates each, then conditioned on a
certain block name `i' appearing, all the pairwise inner products 
are approximately preserved. In other words, even the unitary $U$ is
applied only in the superpositional setting. One may now wonder if
we can implement the unitary $U$ via an efficient quantum circuit
(i.e. of size $\polylog(d_1)$). If so, this would give rise to an
efficient {\em quantum Johnson Lindenstrauss transform}, akin to 
efficient quantum Fourier transforms arising from classical discrete
Fourier transforms (e.g.~\cite{Coppersmith,MZ}). 
The efficient quantum Fourier transform is at the heart of 
many famous quantum algorithms, including Shor's algorithms
for integer factoring and discrete logarithm~\cite{Shor}. 

We show that with high probability, a uniformly random 
$d_1 \times d_1$ unitary from
an approximate $t$-design, where $t = \Theta(d_2)$, suffices for an
efficient quantum Johnson Lindenstrauss transform.
For this value of $t$, both choosing a uniformly random unitary from
the  $t$-design as well as applying it to quantum states are efficient
to implement by quantum algorithms. This follows from the fact that
so-called {\em local random quantum circuits} of size 
$s = t^{10} (\log d_1)^2 \log (1/\alpha)$ form
an $\alpha$-approximate $t$-design of $d_1 \times d_1$ unitaries 
with high probability~\cite{BHH}. The number of random bits required
to describe such a local random circuit is at most 
$O(s \log s \log \log d_1)$.

A limitation of our quantum Johnson Lindenstrauss transform is that 
the distribution over the block
names is almost uniform. We thus have no `control' over the block name,
because of which we cannot apply our transform for most of the classical
settings where the Johnson Lindenstrauss lemma was used (in the classical
explicit setting, one can always force the block to be the first block
without any trouble). Nevertheless, we do give a toy application 
of our transform
to the important problem of private information retrieval. Finding
more applications of our transform is an important open problem.

\paragraph{Related work:}
Our quantum Johnson Lindenstrauss transform approximately preserves the
pairwise inner products for a block name with high probability over
the choice of the unitary from the design.
The block dimension is $d_2$. If one wants to approximately preserve
the pairwise overlaps averaged over all the unitaries from a finite
set, then there much smaller block sizes suffice. This variant is also
known as quantum identification codes. Fawzi, Hayden and 
Sen~\cite{FHS} constructed such codes with very small block size by
efficiently quantising (in the sense of quantum Fourier transform
versus classical discrete Fourier transform) low distortion embeddings 
of $\ell_2$ into $\ell_1$.

Harrow, Montanaro and Short~\cite{HMS} have shown the impossiblity of
obtaining a Johnson Lindenstrauss style dimenionality reduction for
mixed quantum states under the Frobenius norm (aka Schatten $2$-norm).
The impossibility proof uses a feature similar to the observation above
that the block name is essentially uniform.

The Johnson-Lindenstrauss lemma has found several applications in
quantum algorithms and protocols too e.g. quantum 
fingerprinting~\cite{BCWdW,GKdW}, non-local games~\cite{CHTW} etc.

\section{Preliminaries}
\label{sec:preliminaries}
Let $\elltwo{v} := \sqrt{\sum_{i=1}^d |v_i|^2}$ denote the 
$\ell_2$-norm of a vector $v \in \C^{d}$.
Similarly, for a matrix $M \in \C^{d_1} \times \C^{d_2}$, let 
$\elltwo{M}$ denote the Frobenius norm or Hilbert-Schmidt norm or the 
Schatten $2$-norm which
is nothing but the $\ell_2$-norm of the $(d_1 d_2)$-tuple obtained by
stretching $M$ to a long vector.

\subsection{Unitary $t$-designs}
We recall the definition of a tensor product expander (TPE) first
defined by Harrow and Hastings~\cite{HH}.
\begin{definition}[{\bf Tensor product expander}]
A $(d, s, \lambda, t)$-tensor product expander (TPE) is a set of 
$d \times d$ unitaries $\{V_i\}_{i=1}^s$ such that
\[
\elltwo{
\E_V^{\mbox{Design}} [V^{\otimes t} M (V^\dagger)^{\otimes t}] -
\E_{U}^{\mbox{Haar}} [U^{\otimes t} M (U^\dagger)^{\otimes t}]
} \leq
\lambda \elltwo{M},
\]
for all linear operators 
$M: (\C^d)^{\otimes t} \rightarrow (\C^d)^{\otimes t}$.
The notation 
\[
\E_V^{\mbox{Design}} [V^{\otimes t} M (V^\dagger)^{\otimes t}] 
:=
s^{-1} \sum_{i=1}^s 
V_i^{\otimes t} M (V_i^\dagger)^{\otimes t} 
\]
denotes the expectation under the choice of a uniformly random unitary
from the design.
The notation $\E_U^{\mbox{Haar}}[\cdot]$ denotes the expectation under 
the choice of a unitary $U$ picked from the Haar measure.
\end{definition}

We now recall the definition of an approximate unitary $t$-design 
according to Low~\cite{Low}.
\begin{definition}[{\bf Unitary $t$-design}]
Consider $d^2$ formal variables $\{u_{ij}\}_{i,j=1}^d$. A monomial $M$
in these formal variables is said to be {\em balanced of degree $t$} if
it is a product of exactly $t$ of the formal variables and exactly 
$t$ of complex conjugates of the formal variables (the sets
of unconjugated and conjugates variables bear no relation amongst them).
For a $d \times d$ unitary matrix $V$, let $M(V)$ denote the value
of the monomial $M$ obtained by evaluating it at the entries $V_{ij}$ of
$V$. A balanced polynomial of degree $t$ is a linear combination of
balanced monomials of degree $t$.

A unitary $(d, s, \alpha, t)$-design is a set of 
$d \times d$ unitaries $\{V_i\}_{i=1}^s$ such that
\[
\left|
\E_{V}^{\mbox{Design}} [M(V)] -
\E_{U}^{\mbox{Haar}} [M(U)]
\right| \leq
\frac{\alpha}{d^t},
\]
for all balanced monomials $M$ of degree $t$.
\end{definition}

Sequentially iterating a TPE twice means applying the superoperator 
corresponding to the TPE twice in succession. This gives us a 
$(d, s^2, \lambda^2, t)$-TPE where the $s^2$ unitaries are of the form
$V_i V_j$, $1 \leq i, j \leq s$. 
It is now easy to see that a $(d, s, \lambda, t)$-TPE can be 
sequentially iterated 
$O(\frac{t \log d + \log \alpha^{-1}}{\log \lambda^{-1}})$ times to
obtain an $\alpha$-approximate unitary $t$-design. For a 
proof of this statement, we refer to \cite[Lemma~2.7]{Low}.

\subsection{Johnson-Lindenstrauss lemma}
\label{subsec:jl}
We first recall the following well known concentration property 
of the sum of squares of iid Gaussians (aka the chi-square distribution), 
which can be easily proved Chernoff
style using the exponential moment generating function.
\begin{fact}
\label{fact:chisquare}
Let $G_1, \ldots, G_n$ be independent Gaussians of mean $0$ and variance
$1$ each. Let $\epsilon > 0$. Then
\[
\Pr \left[
\sum_{i=1}^{n} G_i^2 \not \in (1 \pm \epsilon) n
\right] \leq
2 (e^{-\epsilon/2} \sqrt{1+\epsilon})^n.
\]
For $\epsilon \leq 1$, we can further upper bound the right hand side by
$
2 (e^{-\epsilon/2} \sqrt{1+\epsilon})^n \leq
2 e^{-2^{-3} \epsilon^2 n}.
$
\end{fact}

We now state the main technical lemma behind the proof of the Johnson
Lindenstrauss lemma which gives a concentration result for the length
of the projection of a unit vector onto a Haar random subspace.
This lemma
can be proved by appealing to Levy's lemma about concentration
of a Lipschitz function defined on the unitary group around its mean,
combined with Fact~\ref{fact:chisquare} above.
\begin{fact}
\label{fact:conc}
Let $v$ be a fixed vector in $\C^{d_1}$, $\elltwo{v} = 1$. Let $d_2 < d_1$.
Let $U$ be a Haar random $d_1 \times d_1$ unitary.
Let $\Pi_i$, $1 \leq i \leq \frac{d_1}{d_2}$ be the orthogonal
projection in $\C^{d_1}$ onto the $i$th block of $d_2$ coordinates.
Let $\epsilon > 0$. Then for any fixed $i$,
\[
\Pr_U \left[
\elltwo{\Pi_i U v} \not \in (1 \pm \epsilon) \sqrt{\frac{d_2}{d_1}}
\right]
\leq
4 \exp(-2^{-4} \epsilon^2 d_2).
\]
\end{fact}
\begin{proof}
By symmetry of the Haar measure, the desired probability is nothing but
the probability that a random unit vector in $\C^{d_1}$ does
not have length $(1 \pm \epsilon) \sqrt{\frac{d_2}{d_1}}$ when
projected onto the first $d_2$ coordinates. Since a Haar random unit
vector $v \in \C^{d_1}$ can be generated by taking $2 d_1$ independent
real Gaussian random variables $\{G_i\}_{i=1}^{2 d_1}$ with 
mean $0$ and variance $1$, forming
a complex $d_1$-tuple out of them and then dividing by the $\ell_2$-norm
of the tuple, we can see that for $0 < \epsilon \leq 1$, 
\begin{eqnarray*}
\lefteqn{
\Pr_U \left[
\elltwo{\Pi_i U v} \not \in (1 \pm \epsilon) \sqrt{\frac{d_2}{d_1}}
\right] 
} \\
& \leq &
\Pr \left[
\sum_{i=1}^{2 d_2} G_i^2 \not \in 2 (1 \pm \frac{\epsilon}{2}) d_2 
\right] +
\Pr \left[
\sum_{i=1}^{2 d_1} G_i^2 \not \in 2 (1 \pm \frac{\epsilon}{2}) d_1 
\right] \\ 
& \leq &
2 \exp(-2^{-4} \epsilon^2 d_2) + 
2 \exp(-2^{-4} \epsilon^2 d_1) 
\;\leq\;
4 \exp(-2^{-4} \epsilon^2 d_2),
\end{eqnarray*}
where we used Fact~\ref{fact:chisquare} in the second to last inequality.

For $\epsilon > 1$, only the upper tail is relevant i.e.
\[
\Pr_U \left[
\elltwo{\Pi_i U v} \not \in (1 \pm \epsilon) \sqrt{\frac{d_2}{d_1}}
\right] = 
\Pr_U \left[
\elltwo{\Pi_i U v} > (1 + \epsilon) \sqrt{\frac{d_2}{d_1}}
\right].
\]
Define a real valued function
$f(U) := \elltwo{\Pi_i U v}$.
Then $f(U)$ is $1$-Lipschitz with respect to
the Frobenius norm on $d_1 \times d_1$ unitary matrices. 
By Levy's lemma~\cite[Corollary~4.4.28]{AGZ},
\[
\Pr[|f(U) - \E[f]| > \delta] \leq 2 \exp(-2^{-2} d_1 \delta^2).
\]
where the probability and expectation are taken over the Haar measure 
on $d_1 \times d_1$ unitaries. 
Now observe by symmetry that
$\E[f(U)^2] = \frac{d_2}{d_1}$. By convexity of the square function,
$\E[f(U)] < \sqrt{\frac{d_2}{d_1}}$. Thus,
\[
\Pr[f(U) > (1 + \epsilon) \sqrt{\frac{d_2}{d_1}}] \leq 
2 \exp(-2^{-2} \epsilon^2 d_2).
\]

This covers both the cases of $\epsilon \leq 1$ and 
$\epsilon > 1$ and so completes the proof.
\end{proof}

The Johnson-Lindenstrauss lemma now follows easily from the above fact.
\begin{fact}[{\bf Johnson Lindenstrauss lemma}]
Consider a set of $n$ vectors $\{v_i\}_{i=1}^n \in \C^{d_1}$. Let
$\epsilon > 0$. Then there is a linear transformation 
$T: \C^{d_1} \rightarrow \C^{d_2}$ where $d_2 = O(\epsilon^{-2} \log n)$
such that $\elltwo{T v_i} \in (1 \pm \epsilon) \elltwo{v_i}$ for all 
$i \in [n]$.
\end{fact}
\begin{proof}
Choose a Haar random $d_1 \times d_1$ unitary $U$. For $v \in \C^{d_1}$,
define $T(v) := \sqrt{\frac{d_1}{d_2}} \Pi_1 U v$. Fact~\ref{fact:conc}
and a union bound on probability now completes the proof.
\end{proof}

\section{An efficient quantum Johnson Lindenstrauss transform}
\label{sec:qjl}
In this section, we show that choosing a $d_1 \times d_1$ unitary 
uniformly at random from an approximate unitary $t$-design, for
$t = \Theta(d_2)$, achieves similar performance as the Haar random unitary
in Fact~\ref{fact:conc}. We prove this
by using the method of Low~\cite{Low}, who in turn adapted
the classical $t$-moment method of Bellare and Rompel~\cite{BR} to the
quantum setting. 
It is also possible to give a 
more direct proof by truncating the exponential
moment generating function, used to show concentration for sums of squares
of independent Gaussians in Fact~\ref{fact:chisquare}, at an
appropriately chosen $\Theta(d_2)$th
power and proving that the 
truncation does not affect the value of the generating function by much.
However the value of $t$ obtained by this method is larger than the value
obtained by using Low's method. Hence we will only give the proof using
Low's method. The proof is deferred to Section~\ref{sec:qjlproof}.
\begin{proposition}
\label{prop:qjl}
Let $v$ be a fixed vector in $\C^{d_1}$, $\elltwo{v} = 1$. 
Let $d_2 < d_1$.
Let $U$ be a unitary chosen
uniformly at random from a $(d_1, s, \lambda, t)$-TPE, for
$t = 2^{-9} \epsilon^2 d_2$,  
$
\lambda = 
(\frac{4\epsilon^2 d_2}{d_1^2})^{t/2} e^{-t/2},
$
and 
$
\log s = O(d_2 \log d_1).
$
Let $\Pi_i$, $1 \leq i \leq \frac{d_1}{d_2}$ be the orthogonal
projection in $\C^{d_1}$ onto the $i$th block of $d_2$ coordinates.
Let $0 < \epsilon < 1$. Then for any fixed $i$,
\[
\Pr_U \left[
\elltwo{\Pi_i U v} \not \in (1 \pm \epsilon) \sqrt{\frac{d_2}{d_1}}
\right]
\leq
2^6 \exp(-2^{-10} \epsilon^2 d_2).
\]
\end{proposition}

We can now define the quantum Johnson Lindenstrauss transform and
prove its main property.
\begin{theorem}
\label{thm:qjl}
Consider a set of $n$ pure states $\{\ket{v_i}\}_{i=1}^n \in \C^{d_1}$,
whose classical descriptions are known a priori. 
Let $0 < \epsilon, \delta < 1/4$. Let 
$d_2 = O(\epsilon^{-2} \log \frac{n d_1}{\delta})$.
Let $U$ be a $d_1 \times d_1$ unitary chosen
uniformly at random from a $(d_1, s, \lambda, t)$-TPE, for
$t = 2^{-9} \epsilon^2 d_2$,  
$
\lambda = 
(\frac{4\epsilon^2 d_2}{d_1^2})^{t/2} e^{-t/2},
$
and 
$
\log s = O(d_2 \log d_1).
$
Suppose we apply $U$ to the given pure state and
measure the name of a block of $d_2$ coordinates i.e. we project onto
the range of $\Pi_j$ for some $j$. 
Let $\ket{v_i(j, U)}$ be the normalised state resulting from $\ket{v_i}$
if the name of the measured block is $j$ i.e.
$
\ket{v_i(j, U)} = \frac{\Pi_j U \ket{v_i}}{\elltwo{\Pi_j U \ket{v_i}}}.
$
Then, with probability at least
$1 - \delta$ over the choice of $U$  
\[
\begin{array}{r c l l}
\elltwo{\Pi_j U \ket{v_i}} 
& \in & (1 \pm \epsilon) \sqrt{\frac{d_2}{d_1}} 
&
~~~~~
\forall i \in [n], j \in [\frac{d_1}{d_2}], \\
\braket{v_i(j, U)}{v_{i'}(j, U)} 
& \in & \braket{v_i}{v_{i'}} \pm 8 \epsilon
&
~~~~~
\forall i, i' \in [n], j \in [\frac{d_1}{d_2}].
\end{array}
\]
\end{theorem}
\begin{proof}
From Proposition~\ref{prop:qjl} and 
the union bound on probability, we see that 
\[
\begin{array}{r c l l}
\elltwo{\Pi_j U \ket{v_i}} 
& \in & (1 \pm \epsilon) \sqrt{\frac{d_2}{d_1}} 
&
~~~~~
\forall i \in [n], j \in [\frac{d_1}{d_2}], \\
\elltwo{\Pi_j U \ket{v_i} - \Pi_j U \ket{v_{i'}}} 
& \in & (1 \pm \epsilon) \elltwo{v_i - v_{i'}} \sqrt{\frac{d_2}{d_1}}
&
~~~~~
\forall i, i' \in [n], j \in [\frac{d_1}{d_2}], \\
\elltwo{\Pi_j U \ket{v_i} - \sqrt{-1} \, \Pi_j U \ket{v_{i'}}} 
& \in & (1 \pm \epsilon) \elltwo{v_i - \sqrt{-1} \, v_{i'}} 
                         \sqrt{\frac{d_2}{d_1}}
&
~~~~~
\forall i, i' \in [n], j \in [\frac{d_1}{d_2}],
\end{array}
\]
with probability at lest $1 - \delta$ over the choice of $U$. Using the
above constraints, we get
\begin{eqnarray*}
\lefteqn{\braket{v_i(j, U)}{v_{i'}(j, U)}} \\
& = &
\left(
\frac{1}{2}
\elltwo{
\frac{\Pi_j U \ket{v_i}}{\elltwo{\Pi_j U \ket{v_i}}} -
\frac{\Pi_j U \ket{v_{i'}}}{\elltwo{\Pi_j U \ket{v_{i'}}}}
}^2 - 
1
\right)  \\
&   &
{} -
\sqrt{-1} \,
\left(
\frac{1}{2}
\elltwo{
\frac{\Pi_j U \ket{v_i}}{\elltwo{\Pi_j U \ket{v_i}}} -
\sqrt{-1}\,\frac{\Pi_j U \ket{v_{i'}}}{\elltwo{\Pi_j U \ket{v_{i'}}}}
}^2 - 
1
\right) \\
& \in &
\left(
\frac{1}{2}
\frac{d_1}{d_2}
\left(
\elltwo{
\Pi_j U \ket{v_i} -
\Pi_j U \ket{v_{i'}}
} \pm 
(4 \epsilon / 3)
(\elltwo{\Pi_j U \ket{v_i}} + \elltwo{\Pi_j U \ket{v_{i'}}})
\right)^2 - 1
\right)  \\
&   &
{} -
\sqrt{-1} \,
\left(
\frac{1}{2}
\frac{d_1}{d_2}
\left(
\elltwo{
\Pi_j U \ket{v_i}
\sqrt{-1}\,\Pi_j U \ket{v_{i'}}
} \pm
(4 \epsilon / 3)
(\elltwo{\Pi_j U \ket{v_i}} + \elltwo{\Pi_j U \ket{v_{i'}}})
\right)^2 - 1
\right) \\
& \in &
\left(
\frac{1}{2}
\frac{d_1}{d_2}
\elltwo{
\Pi_j U \ket{v_i} -
\Pi_j U \ket{v_{i'}}
}^2 - 
1
\right)  \\
&   &
{} -
\sqrt{-1} \,
\left(
\frac{1}{2}
\frac{d_1}{d_2}
\elltwo{
\Pi_j U \ket{v_i}
\sqrt{-1}\,\Pi_j U \ket{v_{i'}}
}^2 - 
1
\right) 
\pm 5 \epsilon \\
& \in &
\left(
\frac{1}{2} \elltwo{\ket{v_i} - \ket{v_{i'}}}^2 - 1
\right)  -
\sqrt{-1} \,
\left(
\frac{1}{2} \elltwo{\ket{v_i} - \sqrt{-1}\,\ket{v_{i'}}}^2 - 1
\right) 
\pm 8 \epsilon \\
& \in &
\braket{v_i}{v_{i'}} \pm 8\epsilon.
\end{eqnarray*}
This completes the proof.
\end{proof}

\section{A toy application}
\label{sec:application}
In this section, we will see a toy application of our quantum Johnson
Lindenstrauss transform to protocols for {\em private information 
retrieval}. In this problem there are two parties, Alice and Bob. 
Alice is given a subset
$S \subseteq [m]$, of size $|S| \leq n$. We work in the regime where
$n$ is very small compared to $m$ viz. $n \ll \frac{\log m}{\log \log m}$.
Bob is given an element $x \in [m]$ and he wants to whether $x$ 
lies in $S$ or not. For this purpose, Bob and Alice follow a two
message communication protocol where Bob first sends a message to Alice,
Alice responds and then Bob makes his conclusion whether $x$ lies in $S$
or not. Bob's conclusion should be correct with probability at least $3/4$.
The privacy requirement is that Bob's message should reveal very littel
information about $x$.

Ideally, we would like the messages to be short and the
computing resources used by Alice and Bob to be polynomial
in $n$ and $\log m$. Is this possible? Yes! There is always the trivial
protocol where Bob says nothing and Alice sends Bob the entire subset
$S$ using $O(n \log m)$ bits. The trivial protocol guarantees perfect
privacy for Bob. 

We now ask if there is a protocol guaranteeing at least approximate privacy
for Bob where Alice communications significantly less. Indeed, when 
$n \ll m$ there is such a protocol based on the following fact proved by
Buhrman, Miltersen, Radhakrishnan and Venkatesh~\cite{BMRV}. 
\begin{fact}
\label{fact:setsystem}
There exists a collection $\{T_1, \ldots, T_m\}$ of subsets of
$[n \log m]$, $|T_i| = O(\log m)$, and for every subset $S \subseteq [m]$,
$|S| \leq n$, a scheme of colouring the set $[n \log m]$ with zero or one,
such that for $x \in S$, at least $0.9$ fraction of elements of $T_x$ are
coloured one, and for $x \not \in S$, at least $0.9$ fraction of 
elements of $T_x$ are coloured zero.
\end{fact}
The above fact suggests the following protocol for private information
retrieval. Bob says nothing. Hence perfect privacy holds for Bob. 
Alice sends $\Theta(n)$ random elements of $[n \log m]$ coloured one. 
Her message length is $O(n \log (n \log m))$ bits. Bob checks if the
intersection of Alice's message with $T_x$ is above a certain constant
If so, he declares that $x \in S$; if not, he declares $x \not \in S$.
A standard Chernoff bound shows that there is a constant gap in the 
probability of Bob declaring $x \in S$ depending on whether $x$ really
lies in $S$ or not. A constant number of parallel repetitions of the
protocol suffices to boost the gap and give a success probability of
at least $0.75$ for Bob.

One may now wonder if Alice's communication can be made even more succint.
Unfortunately, not by much because there is a $\Omega(n)$ lower bound
for Alice's message irrespective of Bob's message length under the 
condition of approximate privacy of Bob, which holds
for the quantum setting too. This can be proved by restricting 
Alice's subset  $S$ to satisfy $S \subseteq [n]$, Bob's element $x$
to satisfy $x \in [n]$ and then applying the privacy-privacy tradeoff
of \cite{JRS} for the set membership problem. Nevertheless, there is 
still a gap between the upper and lower bounds for Alice's message size.

We now ask if we can achieve approximate privacy for Bob, short message
for Alice and make Bob's internal computation efficient. Unfortunately,
the set system guaranteed by Fact~\ref{fact:setsystem} is non-explicit.
Near explicit constructions of similar set systems were later provided
by Ta-Shma~\cite{TaShma} and Capalbo, Reingold, Vadhan and 
Wigderson~\cite{CRVW}, but their
parameters are worse and Bob's internal computation is still not proved
to be efficient.

We now give a quantum protocol achieving approximate privacy for 
Bob, short message for Alice and efficient internal computation for
Bob. Our protocol uses the efficient quantum Johnson-Lindenstrauss 
transform. The idea behind the protocol is as follows. For a subset
$S \subseteq [m]$, define the following pure quantum state 
$\ket{S} := |S|^{-1/2} \sum_{y \in S} \ket{y}$ in $\C^m$. If $x \in S$
$\braket{x}{S} \geq n^{-1/2}$. If $x \not \in S$, $\braket{x}{S} = 0$.
Now suppose we apply the quantum Johnson Lindenstrauss transform
of Theorem~\ref{thm:qjl} with
$\epsilon := 0.01 n^{-3}$ and measure the name of a block, say 
$i \in [\frac{d_1}{d_2}]$, where $d_1 := m$, 
$d_2 = O(\epsilon^{-2} \log n d_1) = O(n^6 \log m)$. 
The unitary $U$ from the 
$(d_1, s, \lambda, t)$-TPE where $t = O(\epsilon^2 d_2)$,
$\lambda = (\frac{4\epsilon^2 d_2}{d_1^2})^{t/2} e^{-t/2}$,
that is chosen by the transform can be described using 
$\log s = O(d_2 \log d_1)$ bits. Moreover, constructing and applying
the quantum circuit to quantum states,  given the name of the unitary,
can be done in time $\poly(n, \log m)$. Let $\ket{x'}$, $\ket{S'}$ be
the resulting normalised projections in the $i$th block of dimension 
$d_2 = O(n^6 \log m)$. Then, if $x \in S$, 
$\braket{x'}{S'} \geq 0.9 n^{-1/2}$; if $x \not \in S$, 
$\braket{x'}{S'} \leq 0.1 n^{-1/2}$. The distribution on the block names
is within $\ell_1$-distance $\epsilon$ from the uniform distribution 
irrespective of the element $x \in [m]$.

This leads naturally to the following quantum protocol for private
information retrieval, where Alice is given $S \subseteq [m]$,
$|S| \leq n$ and Bob is given $x \in [m]$.
\begin{enumerate}

\item
At first, independently of $x$, Bob chooses a uniformly random unitary 
$U$ from the TPE. He then applies $U$
to $\ket{x}$ and measures the name of a block. He stores the collapsed
pure state that lives in the residual $d_2$-dimensional spaces.
He repeats this process
(with the same $U$ and $\ket{x}$) independently $\Theta(n^2)$ times. 
He then sends Alice the description of $U$, which is like a 
{\em public coin}, followed by the $\Theta(n^2)$ 
block names that were measured (note that in general, they are all 
different); 

\item
Alice makes $\Theta(n^2)$ projections of $\ket{S}$ into $d_2$-dimensional
space corresponding to the 
unitary $U$ and the block names received from Bob. She then sends
these $\Theta(n^2)$ pure quantum states to Bob;

\item
Bob performs $\Theta(n^2)$ SWAP tests between the pure states that
Alice sent versus the pure states that he obtained in the first step
above by collapsing. From the results of these tests, he checks 
whether the fraction of successes was larger than 
$\frac{1}{2} + \frac{0.2}{n}$ or not. If yes, he declares that
$x$ lies in $S$. If not, he declares that $x$ does not lie in $S$.

\end{enumerate}

Bob's message is classical and consists of 
$\log s = O(n^6 (\log m)^2)$ bits of public coin
followed by $O(n^2 \log m)$ bits for the block names. 
Bob's internal computation is efficient i.e. takes 
time $\poly(n, \log m)$. 
The public coin
can be reduced to $O(n \log m)$ bits by a standard technique of
Newman~\cite{Newman}, but then Bob's internal computation is no longer
guaranteed to be efficient. 
Bob's message is almost private since the probability distribution on the
block names is at most $O(\epsilon n^2) = O(1/n)$ in $\ell_1$-distance
from uniform. Alice's message is quantum and consists of
$O(n^2 (\log n + \log \log m))$ qubits. For 
$n \ll \frac{\log m}{\log \log m}$, this is less than $O(n \log m)$.
By a standard Chernoff bound, Bob reaches the correct conclusion
whether $x$ lies in $S$ or not with probability at least $3/4$.

\paragraph{Remark:} The efficient quantum identification code of
Fawzi, Hayden and Sen~\cite[Theorem~4.3]{FHS} can also be easilty 
exploited for private information retrieval. In that protocol,
Bob's message is classical and consists of $O(n^2 \log m)$ bits.
Bob's internal computation is efficient. Bob's message is within
$O(1/n)$ in $\ell_1$-distance from the uniform distribution. 
Alice's message is quantum. However, it consists of
$O(n^2 (\log n + \log \log m) \log \log m)$ qubits, which is more than
Alice's message length in the protocol based on the quantum 
Johnson Lindenstrauss transform. The quantum Johnson Lindenstrauss
transform based protocol achieves small number of qubits for Alice by 
trading off a larger number of bits for Bob, keeping Bob's internal
computation efficient.

\section{Proof of Proposition~\ref{prop:qjl}}
\label{sec:qjlproof}
We use Low's method~\cite{Low}. 
Define the real valued function
$
f(U) := 
\elltwo{\Pi_i U v} - \sqrt{\frac{d_2}{d_1}}
$
where $U$ is a $d_1 \times d_1$ unitary matrix.
From Fact~\ref{fact:conc}, for any $\lambda > 0$,
\[
\Pr_U[|f(U)| \geq \lambda] \leq
4 \exp(-2^{-4} \lambda^2 d_1),
\]
where the probability is taken under the Haar measure on $U$.
Combining this with \cite[Lemma~3.3]{Low}, we get
\[
\E_U[(f(U))^{2m}] \leq
4 \left(\frac{2^4 m}{d_1}\right)^m,
\]
where the expectation is taken over the Haar measure on $U$.
Now define the real valued function
$
g(U) := 
\elltwo{\Pi_i U v}^2 - \frac{d_2}{d_1}.
$
Under the Haar measure on $U$, we have
\begin{eqnarray*}
\lefteqn{\E_U[(g(U))^{2m}]} \\
& = &
\E_U \left[
(f(U))^{2m}  
\left(\elltwo{\Pi_i U v} + \sqrt{\frac{d_2}{d_1}}\right)^{2m}
\right] \\
& \leq &
\left(\frac{4 d_2}{d_1}\right)^m 
\Pr_U \left[\elltwo{\Pi_i U v} \leq \sqrt{\frac{d_2}{d_1}}\right]
\E_{U: \elltwo{\Pi_i U v} \leq \sqrt{\frac{d_2}{d_1}}}[(f(U))^{2m}] \\
&      &
{} +
\sum_{i=2}^{\sqrt{\frac{d_1}{d_2}}}
\left(\frac{((i+1)^2 - 1) d_2}{d_1}\right)^{2m} 
\Pr_U \left[
i \sqrt{\frac{d_2}{d_1}} < \elltwo{\Pi_i U v} \leq 
(i+1) \sqrt{\frac{d_2}{d_1}}
\right] \\
& \leq &
\left(\frac{4 d_2}{d_1}\right)^m \E_U[(f(U))^{2m}] +
4 \sum_{i=2}^{\sqrt{\frac{d_1}{d_2}}}
\left(\frac{((i+1)^2 - 1) d_2}{d_1}\right)^{2m} \exp(-2^{-4} i^2 d_2) \\
& \leq &
2^4 (\frac{2^6 m d_2}{d_1^2})^m +
2^4 (\frac{2^6 d_2^2}{d_1^2})^{m} \exp(-2^{-2} d_2), 
\end{eqnarray*}
where we used Fact~\ref{fact:conc} again in the second inequality.

Now suppose we choose $U$ from a $(d_1, s, \lambda, 2m)$ tensor 
product expander instead
of the Haar measure. Since $(g(U)^{2m}$ is a balanced degree $2m$ 
polynomial in the entries of $U$, its expectation under a TPE
must be close to its expectation under the 
Haar measure. More precisely,
\begin{eqnarray*}
\lefteqn{
|
\E^{\mbox{TPE}}_U[(g(U)^{2m}] -
\E^{\mbox{Haar}}_U[(g(U)^{2m}]
|
} \\
& = &
|
\E^{\mbox{TPE}}_U[
(\Tr[\Pi_j U (\ketbra{v} - \frac{\one}{d_1}) U^\dagger \Pi_j^\dagger])^{2m}
] -
\E^{\mbox{Haar}}_U[
(\Tr[\Pi_j U (\ketbra{v} - \frac{\one}{d_1}) U^\dagger \Pi_j^\dagger])^{2m}
] 
| \\
& = &
|
\E^{\mbox{TPE}}_U[
\Tr[
\Pi_j^{\otimes (2m)} U^{\otimes (2m)} 
(\ketbra{v} - \frac{\one}{d_1})^{\otimes (2m)} 
(U^\dagger)^{\otimes (2m)}
]
] \\
&   &
~~~~~~~
{} -
\E^{\mbox{Haar}}_U[
\Tr[
\Pi_j^{\otimes (2m)} U^{\otimes (2m)} 
(\ketbra{v} - \frac{\one}{d_1})^{\otimes (2m)} 
(U^\dagger)^{\otimes (2m)}
]
] 
| \\
& = &
\left|
\Tr \left[
\Pi_j^{\otimes (2m)} 
\left(
\E^{\mbox{TPE}}_U[
U^{\otimes (2m)} 
(\ketbra{v} - \frac{\one}{d_1})^{\otimes (2m)} 
(U^\dagger)^{\otimes (2m)}
] -
\E^{\mbox{Haar}}_U[
U^{\otimes (2m)} 
(\ketbra{v} - \frac{\one}{d_1})^{\otimes (2m)} 
(U^\dagger)^{\otimes (2m)}
] 
\right)
\right]
\right| \\
& \leq &
\elltwo{\Pi_j^{\otimes (2m)}}
\elltwo{
\E^{\mbox{TPE}}_U[
U^{\otimes (2m)} 
(\ketbra{v} - \frac{\one}{d_1})^{\otimes (2m)} 
(U^\dagger)^{\otimes (2m)}
] -
\E^{\mbox{Haar}}_U[
U^{\otimes (2m)} 
(\ketbra{v} - \frac{\one}{d_1})^{\otimes (2m)} 
(U^\dagger)^{\otimes (2m)}
] 
} \\
& \leq &
(d_2)^m \lambda.
\end{eqnarray*}
Recall that $\lambda$ can be made small at an exponential rate by 
simply sequentially iterating the TPE. 

Now observe that for any probability distribution on $U$, by Markov's
inequality,
\[
\Pr_U
\left[
\elltwo{\Pi_i U v} \not \in (1 \pm \epsilon) \sqrt{\frac{d_2}{d_1}}
\right]
\leq
\Pr_U \left[|g(U)| \geq 2 \epsilon \frac{d_2}{d_1}\right]
\leq
\E_U[(g(U))^{2m}] 
\left(\frac{d_1}{2 \epsilon d_2}\right)^{2m},
\]
where $m$ is any positive integer. Thus,
\begin{eqnarray*}
\lefteqn{
\Pr^{\mbox{TPE}}_U \left[
\elltwo{\Pi_i U v} \not \in (1 \pm \epsilon) \sqrt{\frac{d_2}{d_1}}
\right]
} \\
& \leq &
\left(\frac{d_1}{2 \epsilon d_2}\right)^{2m}
\left(
2^4 \left(\frac{2^6 m d_2}{d_1^2}\right)^m +
2^4 \left(\frac{2^6 d_2^2}{d_1^2}\right)^{m} \exp(-2^{-2} d_2) +
d_2^m \lambda 
\right) \\
& = &
2^4 \left(\frac{2^4 m}{\epsilon^2 d_2}\right)^m +
2^4 \left(\frac{2^4}{\epsilon^2}\right)^{m} \exp(-2^{-2} d_2) +
\left(\frac{d_1^2}{4\epsilon^2 d_2}\right)^m \lambda. 
\end{eqnarray*}
Choosing $m := 2^{-10} \epsilon^2 d_2$, we get
\begin{eqnarray*}
\lefteqn{
\Pr^{\mbox{TPE}}_U \left[
\elltwo{\Pi_i U v} \not \in (1 \pm \epsilon) \sqrt{\frac{d_2}{d_1}}
\right]
} \\ 
& \leq &
2^4 \exp(-2^{-10} \epsilon^2 d_2) + 
2^4 \exp(2^{-7} \epsilon^2 \ln (1/\epsilon) d_2) 
\exp(-2^{-2} d_2) + 
\left(\frac{d_1^2}{4\epsilon^2 d_2}\right)^m \lambda \\
& \leq &
2^6 \exp(-2^{-10} \epsilon^2 d_2),
\end{eqnarray*}
by taking 
$
\lambda < 
(\frac{d_1^2}{4\epsilon^2 d_2})^{-m} e^{-2^{-10} \epsilon^2 d_2}.
$

Note that starting from a TPE with constant value of parameter 
$\lambda_0$ and a constant number of unitaries $s$ 
we can sequentially iterate it
\[
k :=
\frac{2 m \log d_1 + 2m \log (1/\epsilon) + 2^{-10} \epsilon^2 d_2}
{\log (1/\lambda_0)} 
\leq
\frac{2^{-8} d_2 \log d_1}{\log (1/\lambda_0)}
\]
times in order to get 
$
\lambda
$
as small as above. 
Existence of $(d, \poly(1/\lambda_0), \lambda_0, t)$-TPEs for constant
$\lambda_0$ and
$d \geq \poly(t)$ was shown by Harrow and Hastings~\cite{HH} via a
probabilistic argument. Efficient constructions of such TPEs for 
$t = \polylog(d)$ was shown by Sen~\cite{Sen} by combining the existence
result of Harrow and Hastings together with the zigzag product for
quantum expanders~\cite{BST}.
For many applications including the one 
to Johnson-Lindenstrauss, the above expression for $k$ 
is polynomial in the input parameters.
Moreover, choosing a uniformly random unitary from such a design
takes only $O(k)$ random bits as opposed to the $(1/\lambda)^{O(d_1^2)}$
random bits required to choose a Haar random unitary to within
Frobenius distance of $\lambda$.

This completes the proof of Proposition~\ref{prop:qjl}.

\section*{Acknowledgements}
I thank Ashley Montanaro for pointing me to his work~\cite{HMS} 
on compression of
quantum states and Johnson Lindenstrauss lemma during a talk given on a
preliminary version of this work at a workshop in CRM, Montr\'{e}al,
Canada, October 2011.

\bibliography{jl}

\end{document}